\documentclass[10pt]{article}

\usepackage[dvipdfmx]{graphicx}
\usepackage{latexsym}
\usepackage{pstricks,psfrag}
\usepackage{amsmath,amssymb}
\usepackage{algorithm}
\usepackage{algpseudocode}
\usepackage{longtable}
\usepackage{url}

\newtheorem{cor}{Corollary}
\newtheorem{prop}{Proposition}
\newtheorem{lem}{Lemma}
\newtheorem{thm}{Theorem}

\newcommand{\figref}[1]{Figure~{\rm\ref{fig:#1}}}
\newcommand{\tabref}[1]{Table~{\rm\ref{tab:#1}}}

\newcommand{\lemref}[1]{Lemma~{\rm\ref{lem:#1}}}
\newcommand{\lineref}[1]{Line~{\rm\ref{line:#1}}}
\newcommand{\propref}[1]{Proposition~{\rm\ref{prop:#1}}}
\newcommand{\secref}[1]{Section~\ref{sec:#1}}
\newcommand{\thmref}[1]{Theorem~\ref{thm:#1}}
\renewcommand{\algref}[1]{Algorithm~\ref{alg:#1}}

\newcommand{\qed}{\quad $\Box$\medskip}

\newenvironment{proof}{\medskip
\noindent{\scshape Proof:}}{\quad $\Box$\medskip}

\long\def\invis#1{}

\title{An Efficient Local Search for the Minimum Independent Dominating Set Problem}
\author{Kazuya Haraguchi}
\date{Otaru University of Commerce, Midori 3-5-21, Otaru, Hokkaido, Japan\\
  \texttt{haraguchi@res.otaru-uc.ac.jp}}

\begin{document}
\maketitle

\begin{abstract}
In the present paper, 
we propose an efficient local search for the minimum independent dominating set problem. 
We consider a local search that uses {\em $k$-swap\/} as 
the neighborhood operation.
Given a feasible solution $S$, it is the operation 
of obtaining another feasible solution
by dropping exactly $k$ vertices from $S$ 
and then by adding any number of vertices to it.
We show that, when $k=2$,
(resp., $k=3$ and a given solution is minimal with respect to 2-swap),
we can find an improved solution in the neighborhood 
or conclude that no such solution exists in $O(n\Delta)$
(resp., $O(n\Delta^3)$) time,
where $n$ denotes the number of vertices and $\Delta$ denotes the maximum degree. 
We develop a metaheuristic algorithm that repeats
the proposed local search and the plateau search iteratively,
where the plateau search examines solutions of the same size as the current solution
that are obtainable by exchanging a solution vertex and a non-solution vertex. 
The algorithm is so effective that, among 80 DIMACS graphs,
it updates the best-known solution size for five graphs
and performs as well as existing methods for the remaining graphs. 
\end{abstract}

\section{Introduction}
Let $G=(V,E)$ be a graph
such that $V$ is the vertex set and $E$ is the edge set. 
Let $n=|V|$ and $m=|E|$. 
A vertex subset $S$ $(S\subseteq V)$ 
is {\em independent} if no two vertices in $S$ are adjacent,
and {\em dominating} 
if every vertex in $V\setminus S$ is adjacent to at least one vertex in $S$. 
Given a graph,
the {\em minimum independent dominating set\/} ({\em MinIDS}) problem
asks for a smallest vertex subset
that is dominating as well as independent. 
The MinIDS problem has many practical applications
in data communication and networks~\cite{KNMW.2005}.

There is much literature on the MinIDS problem
in the field of discrete mathematics~\cite{GH.2013}. 
%
The problem is NP-hard~\cite{GJ.1979}
and also hard even to approximate;
there is no constant $\varepsilon>0$
such that the problem can be approximated
within a factor of $n^{1-\varepsilon}$ in
polynomial time, unless P$=$NP~\cite{H.1993}.

For algorithmic perspective, 
Liu and Song~\cite{LS.2006}
and Bourgeois et al.~\cite{BCEP.2013}
proposed exact algorithms with polynomial space. 
The running times of Liu and Song's algorithms
are bounded by $O^\ast(2^{0.465n})$ and $O^\ast(2^{0.620n})$,
and the running time of Bourgeois et al.'s 
algorithm is bounded by $O^\ast(2^{0.417n})$,
where $O^\ast(\cdot)$ is introduced to ignore 
polynomial factors. 
Laforest and Phan~\cite{LP.2013} proposed an exact algorithm
based on clique partition,
and made empirical comparison
with one of the Liu and Song's algorithms,
in terms of the computation time. 
Davidson et al.~\cite{DBL.2017}
proposed an integer linear optimization model 
for the weighted version of the MinIDS problem
(i.e., weights are given to edges as well as vertices,
and the weight of an edge $vx$ is counted as cost if the edge $vx$ is used to
assign a non-solution vertex $v$ to a solution vertex $x$;
every non-solution vertex $v$ is automatically assigned to
an adjacent solution vertex $x$ such that
the weight of $vx$ is the minimum)
and performed experimental validation for random graphs. 
Recently,
Wang et al.~\cite{WCSY.2017} proposed
a tabu-search based memetic algorithm
and Wang et al.~\cite{WLZY.2017} 
proposed a metaheuristic algorithm
based on GRASP (greedy randomized adaptive search procedure).
They showed their effectiveness on DIMACS instances,
in comparison with CPLEX12.6 and LocalSolver5.5.

A vertex subset $S$ is an IDS
iff it is a maximal independent set
with respect to set-inclusion~\cite{B.1962}.
Then one can readily see that the MinIDS problem is equivalent to
the maximum minimal vertex cover (MMVC) problem and
the minimum maximal clique problem. 
Zehavi~\cite{Z.2017} studied the MMVC problem,
which has applications to wireless ad hoc networks,
from the viewpoint of fixed-parameter-tractability.

For a combinatorially hard problem like the MinIDS problem,
it is practically meaningful to develop
a heuristic algorithm to obtain a nearly-optimal solution
in reasonable time. 
In the present paper,
we propose an efficient local search for the MinIDS problem. 
By the term ``efficient'',
we mean that the proposed local search
has a better time bound than one na\"ively implemented. 
The local search can serve as a key tool of local improvement
in a metaheuristic algorithm,
or can be used in an initial solution generator
for an exact algorithm. 
We may also expect that it is extended to the weighted version
of the MinIDS problem in the future work.

Our strategy is to search for a smallest maximal independent set. 
Hereafter, we may call a maximal independent set
simply a {\em solution\/}. 
In the proposed local search,
we use {\em $k$-swap\/} for the neighborhood operation.
Given a solution $S$,  
$k$-swap refers to the operation of obtaining another solution
by dropping exactly $k$ vertices from $S$
and then by adding any number of vertices to it. 
The {\em $k$-neighborhood of $S$\/}
is the set of all solutions that can be 
obtained by performing $k$-swap on $S$. 
We call $S$ {\em $k$-minimal}
if its $k$-neighborhood
contains no $S'$ such that $|S'|<|S|$.

\invis{
To achieve the efficiency,
we design the data structure upon the one 
employed in Andrade et al.'s local search (ARW-LS for short)
for the maximum independent set problem~\cite{ARW.2012}. 
ARW-LS seeks a ``larger'' maximal independent set efficiently,
and we make use of the data structure
in order to search for a ``smaller'' maximal independent set. 
}

To speed up the local search,
one should search the neighborhood 
for an improved solution as efficiently as possible. 
For this, we propose
$k$-neighborhood search algorithms for $k=2$ and 3. 
When $k=2$ (resp., $k=3$ and a given solution is 2-minimal),
the algorithm finds an improved solution 
or decides that no such solution exists in $O(n\Delta)$
(resp., $O(n\Delta^3)$) time,
where $\Delta$ denotes the maximum degree in the input graph.

Furthermore, we develop a metaheuristic algorithm
named {\em ILPS} ({\em Iterated Local \& Plateau Search}) 
that repeats the proposed local search
and the plateau search iteratively.  
ILPS is so effective that, among 80 DIMACS graphs,
it updates the best-known solution size for five graphs
and performs as well as existing methods
for the remaining graphs.
\invis{
For all the 80 DIMACS graphs,
ILPS can find solutions of the equal sizes or smaller sizes
than the recent GRASP based method~\cite{WLZY.2017}, CPLEX12.6
and LocalSolver5.5.
}

The paper is organized as follows. 
Making preparations in \secref{prel},
we present $k$-neighborhood search algorithms for $k=2$ and 3
in \secref{ls}
and describe ILPS in \secref{ilps}. 
We show computational results in \secref{comp}
and then give concluding remark in \secref{conc}. 
Some proofs and experimental results
are included in the appendix.
The source code of ILPS is written in C++ and
available at \url{http://puzzle.haraguchi-s.otaru-uc.ac.jp/minids/}.

\section{Preliminaries}
\label{sec:prel}

\subsection{Notation and Terminologies}
For a vertex $v\in V$,
we denote by $\deg(v)$ the degree of $v$, 
and by $N(v)$ the set of neighbors of $v$,
i.e., $N(v)=\{u\mid vu\in E\}$. 
For $S\subseteq V$,
we define $N(S)=(\bigcup_{v\in S}N(v))\setminus S$.
We denote by $G[S]$ the subgraph induced by $S$. 
The $S$ is called a {\em $k$-subset} if $|S|=k$. 

Suppose that $S$ is an independent set.
The {\em tightness of $v\notin S$} is the number of
neighbors of $v$ that belong to $S$, i.e., $|N(v)\cap S|$. 
We call the $v$ {\em $t$-tight} 
if its tightness is $t$. 
In particular, a 0-tight vertex is called {\em free\/}.
We denote by $T_t$ the set of $t$-tight vertices. 
Then $V$ is partitioned into $V=S\cup T_0\cup\dots\cup T_{n-1}$,
where $T_t$ may be empty. 
Let $T_{\ge t}$ denote the set of vertices 
that have the tightness no less than $t$,
that is, $T_{\ge t}=T_{t}\cup T_{t+1}\cup\dots\cup T_{n-1}$. 

An independent set $S$ is a solution (i.e., a maximal independent set) iff $T_0=\emptyset$. 
We call $x\in S$ a {\em solution vertex\/}
and $v\notin S$ a {\em non-solution vertex\/}. 
When a solution vertex $x\in S$ and
a $t$-tight vertex $v\notin S$
are adjacent,
$x$ is a {\em solution neighbor of $v$},
or equivalently, $v$ is a {\em $t$-tight neighbor of $x$\/}. 

A $k$-swap on a solution $S$
is the operation of obtaining another solution $(S\setminus D)\cup A$
such that $D$ is a $k$-subset of $S$
and that $A$ is a non-empty subset of $V\setminus S$. 
We call $D$ a {\em dropped subset} and $A$ an {\em added subset\/}. 
The $k$-neighborhood of $S$
is the set of all solutions obtained by performing a $k$-swap on $S$. 
A solution $S$ is $k$-minimal if the $k$-neighborhood
contains no improved solution $S'$ such that $|S'|<|S|$. 
Note that every solution is 1-minimal.

If a $k$-subset $D$ is dropped from $S$,
then trivially, the $k$ solution vertices in $D$ become free,
and some non-solution vertices may also become free. 
Observe that a non-solution vertex becomes free
if the solution neighbors are completely contained in $D$. 
We denote by $F(D)$ the set of such vertices
and it is 
defined as $F(D)=\{v\in V\setminus S\mid N(v)\cap S\subseteq D\}$. 
Clearly the added subset $A$ should be
a maximal independent set in $G[D\cup F(D)]$.
We have $F(D)\subseteq N(D)$, and
the tightness of any vertex in $F(D)$ 
is at most $k$ (at the time before dropping $D$ from $S$).

\subsection{Data Structure}
\label{sec:prel.ds}
We store the input graph by means of 
the typical adjacency list. 
We maintain a solution based on the data structure
that Andrade et al.~\cite{ARW.2012} invented for
the maximum independent set problem. 
For the current solution $S$,
we have an ordering $\pi:V\rightarrow\{1,\dots,n\}$ 
on all vertices in $V$ such that;
\begin{itemize}
  \item $\pi(x)<\pi(v)$ whenever $x\in S$ and $v\notin S$;
  \item $\pi(v)<\pi(v')$ whenever $v\in T_0$ and $v'\in T_{\ge 1}$;
  \item $\pi(v')<\pi(v'')$ whenever $v'\in T_{1}$ and $v''\in T_{\ge 2}$;
  \item $\pi(v'')<\pi(v''')$ whenever $v''\in T_2$ and $v'''\in T_{\ge 3}$.
\end{itemize}
Note that the ordering is partitioned into five sections;
$S$, $T_0$, $T_1$, $T_2$ and $T_{\ge3}$. 
In each section, the vertices are arranged arbitrarily. 
We also maintain the number of vertices in each section
and the tightness $\tau(v)$ 
for every non-solution vertex $v\notin S$. 

Let us describe the time complexities
of some elementary operations. 
We can scan each vertex section in linear time. 
We can pick up a free vertex (if exists) in $O(1)$ time. 
We can drop (resp., add) a vertex $v$ from (resp., to)
the solution in $O(\deg(v))$ time.  
See \cite{ARW.2012} for details. 

Before closing this preparatory section,
we mention the time complexities of
two essential operations.

\begin{prop}
  \label{prop:listFD}
  Let $D$ be a $k$-subset of $S$. 
  We can list all vertices in $F(D)$
  in $O(k\Delta)$ time. 
\end{prop}
\begin{proof}
  We let every $v\in V$ have an integral counter, which we denote by $c(v)$. 
  It suffices to scan vertices in $N(D)$ twice. 
  In the first scan, we initialize the counter value as $c(u)\gets 0$
  for every neighbor $u\in N(x)$ of every solution vertex $x\in D$. 
  In the second, we increase the counter of $u$
  by one (i.e., $c(u)\gets c(u)+1$)
  when $u$ is searched in the adjacency list of $x\in D$.
  Then, if $c(u)=\tau(u)$ holds,
  we output $u$ as a member of $F(D)$ 
  since the equality represents that 
  every solution neighbor of $u$ is contained in $D$. 
  Obviously the time bound is $O(k\Delta)$. 
\end{proof}

\begin{prop}
  \label{prop:adj}
  Let $D$ be a $k$-subset of $S$. 
  For any non-solution vertex $v\in F(D)$,
  we can decide whether $v$ is adjacent to all vertices in $F(D)\setminus\{v\}$
  in $O(k\Delta)$ time. 
\end{prop}
\begin{proof}
  We use the algorithm of \propref{listFD}.
  As preprocessing of the algorithm,
  we set the counter $c(u)$ of each $u\in N(v)$ to 0, i.e., $c(u)\gets0$,
  which can be done in $O(\deg(v))$ time. 
  After we acquire $F(D)$ by running the algorithm of \propref{listFD},
  we can see if $v$ is adjacent to all other vertices in $F(D)$
  in $O(\deg(v))$ time
  by counting the number of vertices $u\in N(v)$
  such that $\tau(u)\in\{1,\dots,k\}$ and $c(u)=\tau(u)$.
  If the number equals to (resp., does not equal to) $|F(D)|-1$,
  then we can conclude that it is true (resp., false). 
\end{proof}

\section{Local Search}
\label{sec:ls}

Assume that, for some $k\ge2$, a given solution $S$ is $k'$-minimal 
for every $k'\in\{1,\dots,k-1\}$. 
Such $k$ always exists, e.g., $k=2$. 
In this section, we consider how we find an improved solution
in the $k$-neighborhood of $S$
or conclude that $S$ is $k$-minimal efficiently. 

Let us describe how time-consuming na\"ive implementation is. 
In na\"ive implementation,
we search all $k$-subsets of $S$ as candidates of the dropped subset $D$,
where the number of them is $O(n^k)$.
Furthermore, for each $D$, 
there are $O(n^{k-1})$ candidates of the added subset $A$. 
The number of possible pairs $(D,A)$
is up to $O(n^{2k-1})$. 

In the proposed neighborhood search algorithm, 
we do not search dropped subsets but added subsets;
we generate a dropped subset from each added subset. 
When $k\in\{2,3\}$,
the added subsets can be searched more efficiently
than the dropped subsets.
This search strategy stems from \propref{ls.general}, 
a necessary condition of a $k$-subset $D$
that the improvement is possible by
a $k$-swap that drops $D$. 
We introduce the condition in \secref{ls.cond}. 

%
%
%
%
Then in \secref{ls.2} (resp., \ref{sec:ls.3}),
we present a $k$-neighborhood search algorithm
that finds an improved solution
or decides that no such solution exists
for $k=2$ (resp., 3),
which runs in $O(n\Delta)$
(resp., $O(n\Delta^3)$) time.

\subsection{A Necessary Condition for Improvement}
\label{sec:ls.cond}

Let $D$ be a $k$-subset of $S$. 
If there is a subset $A\subseteq F(D)$
such that $A$ is maximal independent in $G[D\cup F(D)]$
and $|A|<|D|$, then we have an improved solution $(S\setminus D)\cup A$. 
The connectivity of $G[D\cup F(D)]$
is necessary for the existence of such $A$, 
as stated in the following proposition. 

\begin{prop}
\label{prop:ls.general}
Suppose that a solution $S$ is $k'$-minimal
for every $k'\in\{1,\dots,k-1\}$ for some integer $k\ge2$.
Let $D$ be a $k$-subset of $S$. 
There is a maximal independent set $A$ 
in $G[D\cup F(D)]$ such that $A\subseteq F(D)$ and $|A|<|D|$
only when the subgraph is connected.  
\end{prop}
\begin{proof}
Suppose that $G[D\cup F(D)]$ is not connected. 
Let $q$ be the number of connected components
and $D^{(p)}\cup F^{(p)}(D)$ be the subset of vertices in the $p$-th component
($q\ge2$, $p=1,\dots,q$, $D^{(p)}\subseteq D$, $F^{(p)}(D)\subseteq F(D)$). 
Each $D^{(p)}$ is not empty 
since otherwise 
there would be an isolated vertex in $F^{(p)}(D)$.
It is a free vertex with respect to $S$, 
which contradicts that $S$ is a solution.
Then we have $1\le|D^{(p)}|<k$. 

The maximal independent set $A$ is a subset of $F(D)$. 
We partition $A$ into $A=A^{(1)}\cup\dots\cup A^{(q)}$, where $A^{(p)}=A\cap F^{(p)}(D)$. 
Each $A^{(p)}$ is maximal independent for the $p$-th component.  
As $|A|<|D|$, $|A^{(p)}|<|D^{(p)}|$ holds for some $p$.  
Then we can construct an improved solution $(S\setminus D^{(p)})\cup A^{(p)}$,
which contradicts the $k'$-minimality of $S$. 
\end{proof}

\subsection{2-Neighborhood Search}
\label{sec:ls.2}

Applying \propref{ls.general} to the case of $k=2$,
we have the following proposition. 


\begin{prop}
\label{prop:ls.2}
Let $D$ be a 2-subset of $S$.
There is a non-solution vertex $v$ in $F(D)$
such that $(S\setminus D)\cup\{v\}$ is a solution
only when there is a 2-tight vertex in $F(D)$. 
\end{prop}
\invis{
\begin{proof}
  Every vertex in $F(D)$ is either 1-tight or 2-tight. 
  If all vertices in $F(D)$ are 1-tight, then 
  the subgraph $G[D\cup F(D)]$ would not be connected 
  since every vertex in $F(D)$ is connected to exactly one vertex in $D$.
\end{proof}
}

We can say more on \propref{ls.2}.
The vertex $v$ should be 2-tight
since, if not so (i.e., $v$ is 1-tight),
$\{v\}$ would not be maximal independent for $G[D\cup F(D)]$;
$v$ is adjacent to only one of $D=\{x,y\}$
from the definition of 1-tightness. 

In summary, if there is an improved solution $(S\setminus D)\cup\{v\}$,
then $v$ is 2-tight and has $x$ and $y$ as the solution neighbors. 
Instead of searching all 2-subsets of $S$,
we scan all 2-tight vertices,
and for each 2-tight vertex $v$,
we take $D=\{x,y\}$ as the candidate of the dropped set. 
We have the following theorem. 

\invis{
\begin{cor}
  \label{cor:ls.2}
  There are $O(n)$ 2-subsets $D$
  for which there exists $v\in F(D)$
  such that $(S\setminus D)\cup\{v\}$
  is a solution. 
\end{cor}
}

\begin{thm}
\label{thm:ls.2}
Given a solution $S$,
we can find an improved solution in the 2-neighborhood
or conclude that $S$ is 2-minimal
in $O(n\Delta)$ time. 
\end{thm}
\begin{proof}
Since we maintain the solution by means of the vertex ordering,
we can scan all the 2-tight vertices in $O(|T_2|)$ time. 
For each 2-tight $v$, we can detect the two solution neighbors,
say $x$ and $y$, in $O(\deg(v))$ time. 

Let $D=\{x,y\}$. The singleton $\{v\}$ is maximal independent for $G[D\cup F(D)]$
and thus we have an improved solution $(S\setminus D)\cup\{v\}$
iff $v$ is adjacent to all other vertices in $F(D)$. 
Whether $v$ is adjacent to all other vertices in $F(D)$ is decided in $O(\Delta)$ time, as we stated in \propref{adj}.
\invis{
Recall the $O(\Delta)$-time algorithm in \propref{listFD}
that enumerates all vertices of $F(D)$. 
As a preprocessing of the algorithm,
we set the counter $c(u)$ of each $u\in N(v)$ to 0, i.e., $c(u)\gets0$,
which can be done in $O(\deg(v))$ time. 
After we run the algorithm of \propref{listFD},
we can see if $v$ is adjacent to all other vertices in $F(D)$ in $O(\deg(v))$ time
by deciding whether or not 
the number of vertices $u\in N(v)$
such that $\tau(u)\in\{1,2\}$ and $c(u)=\tau(u)$
equals to $|F(D)|-1$.
}
If it is the case,
then we can construct an improved solution $(S\setminus D)\cup\{v\}$
in $O(\deg(x)+\deg(y)+\deg(v))=O(\Delta)$ time
as the vertex ordering takes $O(\deg(x))$ time to drop $x$ from $S$
and $O(\deg(v))$ time to add $v$ to it~\cite{ARW.2012}. 
Otherwise, we can conclude that
$(S\setminus D)\cup\{v\}$
is not a solution because some vertices in $F(D)$ are not dominated. 

We have seen that, for each 2-tight vertex $v$, 
it takes $O(\Delta)$ time to find an improved solution $(S\setminus D)\cup\{v\}$
or to conclude that it is not a solution.
Therefore, the overall running time is bounded by $O(|T_2|\Delta)=O(n\Delta)$. 
\end{proof}

\subsection{3-Neighborhood Search}
\label{sec:ls.3}

We have the following proposition 
by applying \propref{ls.general} to the case of $k=3$. 

\begin{prop}
\label{prop:ls.3}
Suppose that $S$ is a 2-minimal solution and that $D=\{x,y,z\}$ is a
$3$-subset of $S$. 
There is a subset $A$ of $F(D)$ such that 
$A$ is maximal independent 
in $G[D\cup F(D)]$ and $|A|<|D|$
only when either of the followings holds:
\begin{description}
\item[(a)] there is a 3-tight vertex in $F(D)$
  that has $x$, $y$ and $z$ as the solution neighbors;
\item[(b)] there are two 2-tight vertices in $F(D)$
  such that one has $x$ and $y$ as the solution neighbors
  and the other has $x$ and $z$ as the solution neighbors.
\end{description} 
\end{prop}

Let us make observation on the added subset. 
Suppose that, for an arbitrary 3-subset $D\subseteq S$,
there is $A\subseteq F(D)$ such that 
$A$ is maximal independent in $G[D\cup F(D)]$ 
and $|A|<|D|$.
When $|A|=1$, the only vertex in $A$ is 3-tight
since otherwise some vertex in $D$ would not be dominated. 
When $|A|=2$, at least one of the two vertices in $A$
is either 2-tight or 3-tight;
if both are 1-tight, one vertex of $D$ would not be dominated.
Concerning the tightness,
the following four situations are possible:
\begin{description}
\item[(i)] $A=\{a\}$ and $a$ is 3-tight;
\item[(ii)] $A=\{a,b\}$, $a$ is 3-tight, and
  $b$ is $t$-tight such that $t\in\{1,2,3\}$;
\item[(iii)] $A=\{a,b\}$, $a$ is 2-tight, and
  $b$ is 2-tight;
\item[(iv)] $A=\{a,b\}$, $a$ is 2-tight, and
  $b$ is 1-tight.
\end{description}
From (ii) to (iv), the vertices $a$ and $b$ are not adjacent.  
We illustrate (i) to (iv) in \figref{N3}. 

\begin{figure}[t!]
  \centering
  \begin{tabular}{cccc}
    \includegraphics[width=3cm,keepaspectratio]{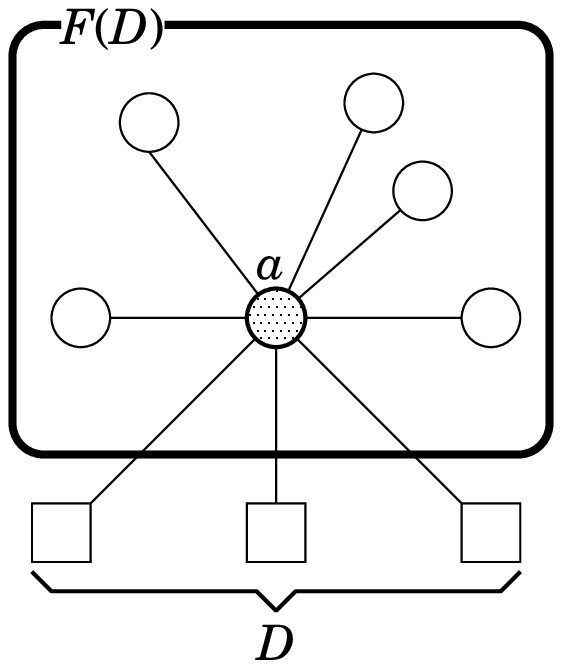} &
    \includegraphics[width=3cm,keepaspectratio]{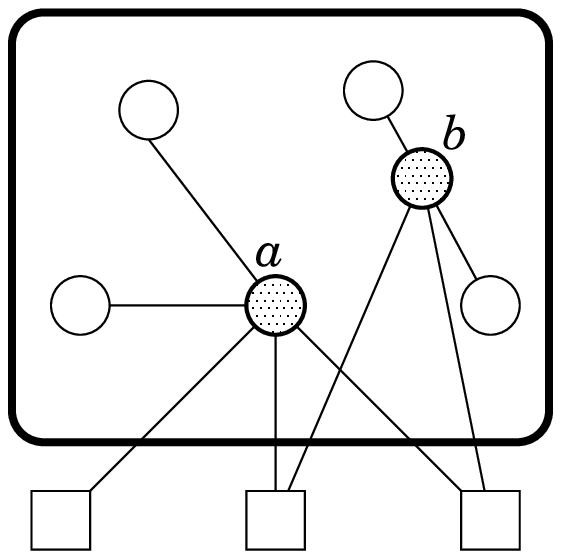} &
    \includegraphics[width=3cm,keepaspectratio]{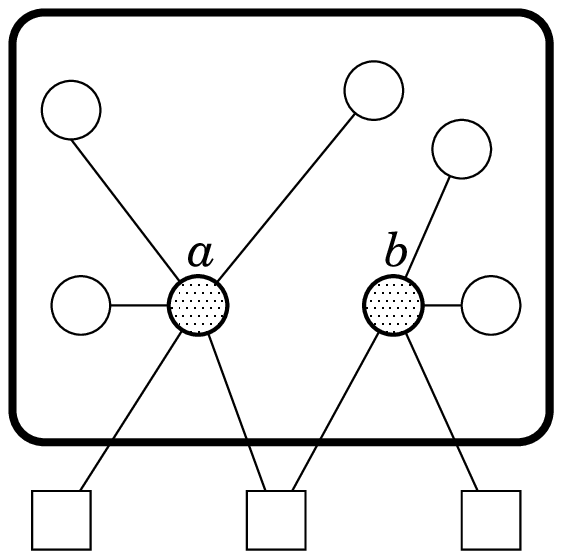} &
    \includegraphics[width=3cm,keepaspectratio]{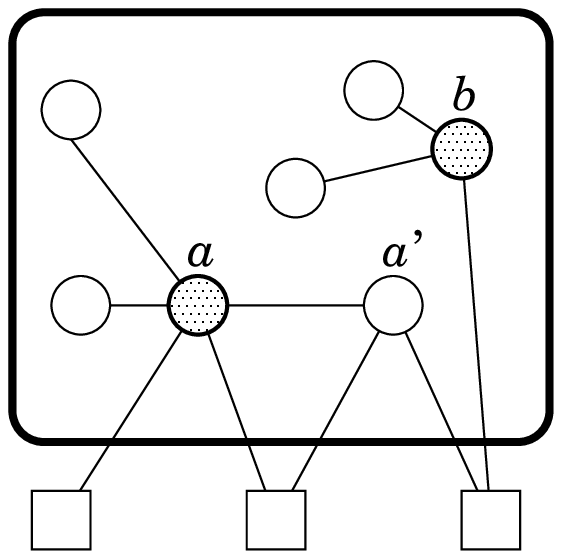} \\
    {\bf(i)} $A=\{a\}$ & {\bf(ii)} $A=\{a,b\}$ &
    {\bf(iii)} $A=\{a,b\}$ & {\bf(iv)} $A=\{a,b\}$
  \end{tabular}
\invis{
  \caption{Illustration of a dropped set $D$ 
    and an added set $A$ for (i) to (iv) in \secref{ls.3}:
    For clarity of the figure,
    we draw only edges that are incident to the vertices $a$ and $b$. 
    Note that every vertex in $F(D)$ is adjacent to at least one vertex in $D$.}
}
  \caption{Illustration of a dropped set $D$ 
   and an added set $A$ for (i) to (iv) in \secref{ls.3}:
    For clarity of the figure,
    we draw only edges that are incident to the vertices $a$, $b$ and $a'$,
    where the vertex $a'$ appears in \lemref{ls.3-4} in the appendix. 
    Note that every vertex in $F(D)$ is adjacent to at least one vertex in $D$.}
    \label{fig:N3}
\end{figure}

Based on the above, we summarize the search strategy as follows. 
In order to generate all 3-subsets $D$ of $S$ such that
$F(D)$ satisfies either (a) or (b) of \propref{ls.3},
we scan all 3-tight vertices $u$ (\propref{ls.3}~(a))
and all pairs of 2-tight vertices, say $v$ and $w$,
such that $|(N(v)\cup N(w))\cap S|=3$ (\propref{ls.3}~(b)). 
For (a), we take $D=N(u)\cap S$
and search $F(D)$ for a 1- or 2-subset $A$ that is maximal independent in $G[D\cup F(D)]$,
regarding the 3-tight vertex $u$ 
as the vertex $a$ in (i) and (ii). 
Similarly, for (b), we take $D=(N(v)\cup N(w))\cap S$
and search $F(D)$ for a 2-subset $A$ that is maximal independent in $G[D\cup F(D)]$,
regarding the 2-tight vertex $v$
as the vertex $a$ in (iii) and (iv). 

We have the following theorem on the time complexity of 3-neighborhood search.
The proof is included in the appendix. 

\begin{thm}
\label{thm:ls.3}
Given a 2-minimal solution $S$,
we can find an improved solution in the 3-neighborhood
or conclude that $S$ is 3-minimal
in $O(n\Delta^3)$ time. 
\end{thm}

\section{Iterated Local \& Plateau Search}
\label{sec:ilps}

In this section, we present a metaheuristic algorithm
named ILPS (Iterated Local \& Plateau Search)
that repeats the proposed local search
and the plateau search iteratively.

We show the pseudo code of ILPS in \algref{ILSPS}.
The ILPS has four parameters, that is $S$, $k$, $\delta$ and $\nu$,
where $S$ is an initial solution, $k$ is the order of the local search
(i.e., a $k$-minimal solution is searched by $\textsc{LocalSearch}(S,k)$
in \lineref{ls}), 
and $\delta$ and $\nu$ are integers.
The roles of the last two parameters
are mentioned in \secref{ilps.penalty}.

The $\textsc{LocalSearch}(S,k)$ in \lineref{ls} is the subroutine
that returns a $k$-minimal solution from an initial solution $S$,
where $k$ is set to either two or three. 
When $k=2$, it determines a 2-minimal solution
by moving to an improved solution repeatedly 
as long as the 2-neighborhood search algorithm delivers one. 
When $k=3$, it first finds a 2-minimal solution,
and then runs the 3-neighborhood search algorithm. 
If an improved solution is delivered, then
the local search moves to the improved solution
and seeks a 2-minimal one again
since the solution is not necessarily 2-minimal. 
Otherwise, the current solution is 3-minimal. 

Below we explain two key ingredients:
the plateau search and the vertex penalty. 
We describe these in Sections \ref{sec:ilps.plateau} and \ref{sec:ilps.penalty}
respectively. 
We remark that they are inspired by {\em Dynamic Local Search\/}
for the maximum clique problem~\cite{PH.2006}
and {\em Phased Local Search\/}
for the unweighted/weighted maximum independent set
and minimum vertex cover~\cite{P.2009}.

\begin{algorithm}[t!]
  \caption{Iterated Local \& Plateau Search (ILPS)}
  \label{alg:ILSPS}
  \begin{algorithmic}[1]
    \Function{ILPS}{$S,k,\delta,\nu$}
    \State $S^\ast\gets S$
    \Comment{$S^\ast$ is used to store the incumbent solution}
    \State $\rho\gets$~a penalty function
    such that $\rho(v)=0$ for all $v\in V$
    \label{line:rhoinit}
    \State $\rho\gets\textsc{UpdatePenalty}(S,\rho,\delta)$
    \label{line:rho1}
    \While{termination condition is not satisfied}
    \State $S\gets\textsc{LocalSearch}(S,k)$
    \label{line:ls}
    \Comment{The local search returns a $k$-minimal solution}
    \State $S\gets\textsc{PlateauSearch}(S,k)$
    \label{line:plateau}
    \Comment{The plateau search returns a $k$-minimal solution}
    \If{$|S|\le|S^\ast|$}
    \State $S^\ast\gets S$
    \EndIf
    \State $S\gets\textsc{Kick}(S^\ast,\rho,\nu)$
    \label{line:kick}
    \Comment{The initial solution of the next iteration is generated}
    \State $\rho\gets\textsc{UpdatePenalty}(S,\rho,\delta)$
    \label{line:rho2}
    \Comment{The penalty function is updated}
    \EndWhile
    \State {\bf return} $S^\ast$
    \EndFunction
  \end{algorithmic}
\end{algorithm}

\subsection{Plateau Search}
\label{sec:ilps.plateau}
In the plateau search (referred to as $\textsc{PlateauSearch}(S,k)$
in \lineref{plateau}),
we search solutions of the size $|S|$
that can be obtained by swapping a solution vertex $x\in S$
and a non-solution vertex $v\notin S$.
Let ${\mathcal P}(S)$ be the collection of all solutions
that are obtainable in this way.
The size of any solution in ${\mathcal P}(S)$ is $|S|$.  
We execute $\textsc{LocalSearch}(S',k)$
for every solution $S'\in{\mathcal P}(S)$,
and if we find an improved solution $S''$ such that $|S''|<|S'|=|S|$,
then we do the same for $S''$, 
i.e., we execute $\textsc{LocalSearch}(P,k)$ for every solution $P\in{\mathcal P}(S'')$. 
We repeat this until no improved solution is found
and employ a best solution among those searched 
as the output of the plateau search.

We emphasize the efficiency of the plateau search;
all solutions in ${\mathcal P}(S)$
can be listed in $O(|T_1|\Delta)$ time. 
Observe that $(S\setminus\{x\})\cup\{v\}$ is a solution
iff  $v$ is 1-tight such that 
$x$ is the only solution neighbor of $v$,
and $v$ is adjacent to all vertices in $F(\{x\})$ other than $v$.
We can scan all 1-tight vertices in $O(|T_1|)$ time.
For each 1-tight vertex $v$,
the solution neighbor $x$ is detected in $O(\deg(v))$ time,
and whether the last condition is satisfied or not is identified 
in $O(\Delta)$ time from \propref{adj}. 
Dropping $x$ from $S$ and adding $v$ to $S\setminus\{x\}$
can be done in $O(\Delta)$ time.

\subsection{Vertex Penalty}
\label{sec:ilps.penalty}
In order to avoid the search stagnation,
one possible approach is to
apply a variety of initial solutions. 
To realize this, we introduce a penalty function
$\rho:V\rightarrow\mathbb{Z}^+\cup\{0\}$ on the vertices.
The penalty function $\rho$ is initialized so that $\rho(v)=0$ for all $v\in V$ (\lineref{rhoinit}). 
During the algorithm, $\rho$ is managed by the subroutine \textsc{UpdatePenalty}
(Lines \ref{line:rho1} and \ref{line:rho2}).  
When the initial solution $S$ of the next local search is determined, 
it increases the penalty $\rho(v)$ of every vertex $v\in S$ by one,
i.e., $\rho(v)\gets\rho(v)+1$.
Furthermore, to ``forget'' the search history long ago,
it reduces $\rho(v)$ to $\lfloor\min\{\rho(v),\delta\}/2\rfloor$
for all $v\in V$ in every $\delta$ iterations.
This $\delta$ is the third parameter of ILPS
and called the {\em penalty delay\/}.

The $\rho$ is used in the subroutine \textsc{Kick} (\lineref{kick}),
the initial solution generator,
so that vertices with fewer penalties are more likely to be included 
in the initial solution. 
\textsc{Kick} generates an initial solution
by adding non-solution vertices (with respect to the incumbent solution $S^\ast$)
``forcibly'' to $S^\ast$.
The added vertices are chosen one by one as follows;
in one trial, \textsc{Kick} picks up one non-solution vertex. 
It then goes on to the next trial with the probability $(\nu-1)/\nu$
or stops the selection with the probability $1/\nu$,
where $\nu$ is the fourth parameter of ILPS.
Observe that $\nu$ specifies the expected number of added vertices. 
In the first trial, \textsc{Kick} randomly picks
up a non-solution vertex that has the fewest penalty. 
In a subsequent $r$-th trial $(r=2,3,\dots)$,
let $R=\{v_1,\dots,v_{r-1}\}$ be the set of vertices chosen so far. 
\textsc{Kick} samples three vertices randomly from
$V\setminus(S^\ast\cup R\cup N(R))$,
and picks up the one that has the fewest penalty among the three. 
Suppose that 
$R=\{v_1,\dots,v_r\}$ has been picked up 
as the result of $r$ trials. 
Then we construct an independent set $S=(S^\ast\setminus N(R))\cup R$.
The $S$ may not be a solution as there may remain free vertices.
If so, we repeatedly pick up free vertices by
the maximum-degree greedy method until $S$ becomes a solution. 
We use the acquired $S$
as the initial solution of the next local search. 

\section{Computational Results}
\label{sec:comp}

We report some experimental results in this section. 
In \secref{comp.phase}, to gain insights into what kind of instance is difficult,
we examine the phase transition
of difficulty with respect to the edge density. 
The next two subsections are devoted to
observation on the behavior of the proposed method.
In \secref{comp.single}, 
we show how a single run of $\textsc{LocalSearch}(S,k)$
improves a given initial solution. 
In \secref{comp.delay}, we show how the penalty delay $\delta$ affects the search.
Finally in \secref{comp.perform}, 
we compare ILPS with the memetic algorithm~\cite{WCSY.2017},
GRASP+PC~\cite{WLZY.2017},
CPLEX12.6~\cite{CPLEX} and LocalSolver5.5~\cite{LSSOL}
in terms of the solution size,
using DIMACS graphs. 

All the experiments are conducted on a workstation that carries
an Intel Core i7-4770 Processor (up to 3.90GHz by means of Turbo Boost Technology)
and 8GB main memory. The installed OS is Ubuntu 16.04.
Under this environment, 
it takes 0.25\,s, 1.54\,s and 5.90\,s approximately to execute {\tt dmclique\/} 
(\url{http://dimacs.rutgers.edu/pub/dsj/clique/})
for instances {\tt r300.5.b}, {\tt r400.5.b} and {\tt r500.5.b}, respectively. 
The ILPS algorithm is implemented in C++ and
compiled by the g++ compiler (ver.~5.4.0) with {\tt -O2} option. 

\subsection{Phase Transition of Difficulty}
\label{sec:comp.phase}
The phase transition 
has been observed for many combinatorial problems~\cite{GW.1994,GS.1997,GS.2002}. 
Roughly, it is said that over-constrained and under-constrained instances
are relatively easy, and that intermediately constrained ones
tend to be more difficult. 

In the MinIDS problem, the amount of constraints
is proportional to the edge density $p$. 
We examine the change of difficulty with respect to $p$. 
We estimate the difficulty of an instance by
how long CPLEX12.8 takes to solve it.

For each $(n,p)\in\{100,150,200\}\times\{0.00,0.05,\dots,1.00\}$,
we generate 100 random graphs (Erd\"os-R\'enyi model)
with $n$ vertices
and the edge density $p$,
i.e., an edge is drawn between two vertices with probability $p$. 
We solve the 100 instances
by CPLEX12.8 and take the averaged computation time. 
We set the time limit of each run to 60\,s.
If CPLEX12.8 terminates by the time limit,
then we regard the computation time as 60\,s. 

\figref{phasetrans} shows the result.
We may say that instances with the edge densities from 0.1 to 0.4
are likely to be more difficult than others.
In fact, the experiments in \cite{DBL.2017,LP.2013}
mainly deal with random graphs with the edge densities in this range.

\begin{figure}[t]
  \centering
  \includegraphics[width=6.6cm,keepaspectratio]{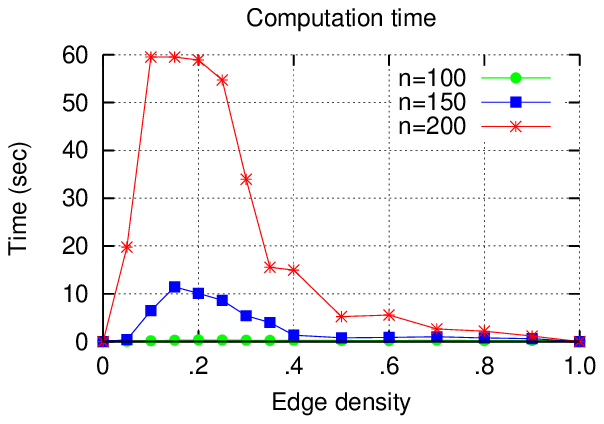}
  \caption{Computation time of CPLEX12.8 for random graphs}
  \label{fig:phasetrans}
\end{figure}

\subsection{A Single Run of Local Search}
\label{sec:comp.single}
We show how a single run of $\textsc{LocalSearch}(S,k)$
improves an initial solution $S$. 
Again we take a random graph. 
We fix the number $n$ of vertices to $10^3$.
For every $p\in\{0.1,\dots,0.9,0.95,0.99\}$, 
we generate $100$ random graphs. 
Then for each graph, we run $\textsc{LocalSearch}(S,k)$ five times,
where we use different random seeds in each time
and construct the initial solution $S$ randomly. 

We show the averaged sizes of random, 2-minimal and 3-minimal solutions
in \tabref{single}. 
We see that, the larger the edge density $p$ is,
the fewer the solution size becomes.
The local search improves a random solution to some extent.
$\textsc{LocalSearch}(S,3)$ improves the solution more than $\textsc{LocalSearch}(S,2)$. 
The difference between the two local searches is
the largest when $p=0.1$, that is $37.37-35.44=1.93$.
The difference gets smaller when $p$ gets larger.
In particular, when $p>0.9$, we see no difference.  

\begin{table}[t]
\centering
\caption{Averaged sizes of random, 2-minimal and 3-minimal solutions in random graphs with $10^3$ vertices}
\label{tab:single}
\begin{tabular}{crrrrrrrrrrr}
\hline
& \multicolumn{1}{c}{$p=.1$} &
\multicolumn{1}{c}{.2} &
\multicolumn{1}{c}{.3} &
\multicolumn{1}{c}{.4} &
\multicolumn{1}{c}{.5} &
\multicolumn{1}{c}{.6} &
\multicolumn{1}{c}{.7} &
\multicolumn{1}{c}{.8} &
\multicolumn{1}{c}{.9} &
\multicolumn{1}{c}{.95} &
\multicolumn{1}{c}{.99} \\
\hline
random & 44.57	&	24.42	&	16.70	&	12.50	&	9.66	&	7.70	&	6.12	&	4.84	&	3.62	&	3.00	&	2.12\\
2-minimal & 37.37	&	20.36	&	13.84	&	10.18	&	7.86	&	6.12	&	4.95	&	3.95	&	2.99	&	2.00	&	1.95\\
3-minimal &35.44	&	19.04	&	12.74	&	9.28	&	7.01	&	5.64	&	4.06	&	3.02	&	2.15	&	2.00	&	1.95\\
\hline
\end{tabular}
\end{table}

Let us discuss computation time. 
In the left of \figref{lsorder},
we show how 
the averaged computation time changes with respect to $p$. 
We see that the computation time of $\textsc{LocalSearch}(S,3)$ is 
tens to thousands of times the computation time of $\textsc{LocalSearch}(S,2)$. 
However, it does not necessarily
diminish the value of the 3-neighborhood search. 
As will be shown in \secref{comp.perform},
when $k=3$, ILPS can find such a good solution
that is not obtained by $k=2$. 

In general, for a fixed $k$,
it takes more computation time when $p$ is larger.
Recall \thmref{ls.2} (resp., \ref{thm:ls.3});
when $k=2$ (resp., 3),
the $k$-neighborhood search algorithm 
finds an improved solution for the current solution $S$
or concludes that $S$ is $k$-minimal in 
$O(n\Delta)$ (resp., $O(n\Delta^3)$) time. 
Roughly, $\Delta$ is increasing as $p$ gets larger. 

For $k=3$, we attribute the peak at $p=0.8$
to the number of 3-tight vertices. 
In the right of \figref{lsorder},
We show the averaged numbers of 2- and 3-tight vertices
with respect to 3-minimal solutions. 
The 3-neighborhood search algorithm searches
2- and 3-tight vertices.
The numbers of both vertices are generally non-decreasing
from $p=0.1$ to 0.8, but when $p>0.8$, 
the number of 3-tight vertices decreases dramatically. 
This is due to the solution size. 
The solution size gives an upper bound
on the tightness of any non-solution vertex,
and when $p>0.8$, the averaged size of 
a 3-minimal solution is less than three; see \tabref{single}.
Since most of the non-solution vertices are either 1- or 2-tight,
we hardly handle the situations (i) and (ii) in \secref{ls.3}.


\invis{
As an extreme example, suppose that $r=1.0$,
i.e., $G$ is a complete graph.
For any initial solution that consists of one vertex,
the $k$-neighborhood search algorithm terminates immediately,
concluding that the solution is $k$-minimal,
for both $k=2$ and 3. 
This is because, in any solution to a complete graph,
all vertices are 1-tight except one solution vertex,
and the 2- and 3-neighborhood search algorithms halt
without scanning any vertex. 
}

\begin{figure}[t]
  \centering
  \begin{tabular}{cc}
    \includegraphics[width=6.6cm,keepaspectratio]{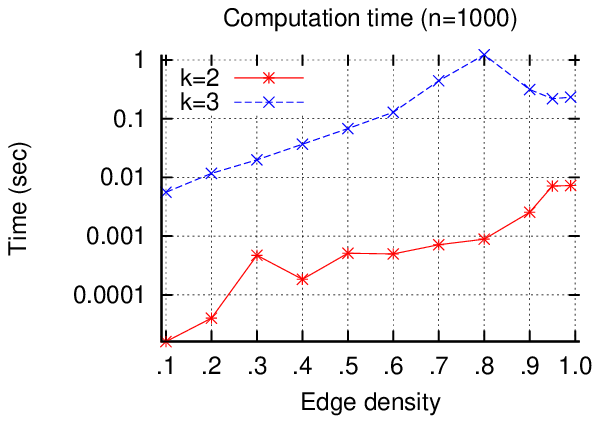} &
    \includegraphics[width=6.6cm,keepaspectratio]{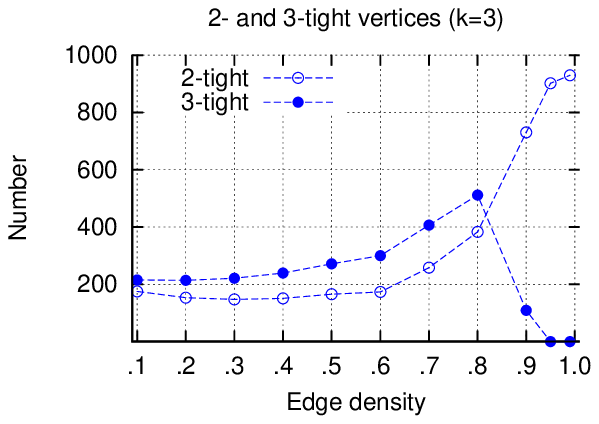}
  \end{tabular}
  \caption{(Left) averaged computation time 
    that $\textsc{LocalSearch}(S,k)$ takes to decide a $k$-minimal solution
    (Right) numbers of 2- and 3-tight vertices with respect to 3-minimal solutions}
  \label{fig:lsorder}
\end{figure}

\subsection{Penalty Delay}
\label{sec:comp.delay}
We introduced the notion of vertex penalty
to control the search diversification.
When the penalty delay $\delta$ is larger,
more varieties of initial solutions are expected to be tested in ILPS. 

To illustrate the expectation,
we evaluate how many iterations ILPS takes
until all vertices are covered by the initial solutions,
that is, used in the initial solutions at least once.
The solid line in \figref{delay} shows the number of iterations
taken to cover all vertices.
The graph we employ here is a $10\times10$ grid graph such that
each vertex is associated with a 2D integral point
$(i,j)\in\{1,\dots,10\}^2$, and that
two vertices $(i,j)$ and $(i',j')$ are adjacent iff $|i-i'|+|j-j'|=1$.
For each $\delta$, 
the number of iterations is averaged over 500 runs of ILPS
with different random seeds, where we fix $(k,\nu)=(2,1)$
and construct the first initial solution $S$
by the maximum-degree greedy algorithm. 

\begin{figure}[t]
  \centering
  \includegraphics[width=6.6cm,keepaspectratio]{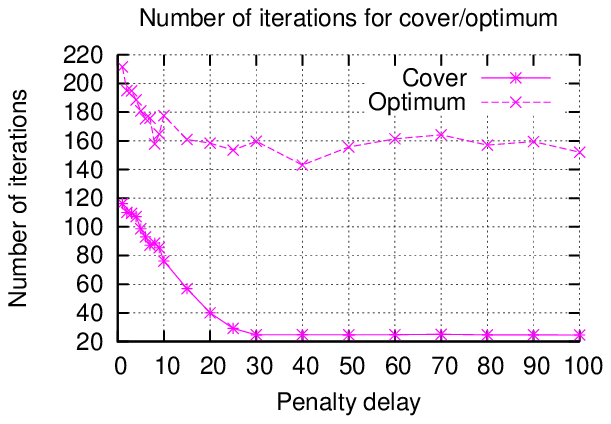}
  \caption{Averaged numbers of iterations to cover all vertices (solid line) and to find the optimum (dashed line)}
  \label{fig:delay}
\end{figure}

The observed phenomenon meets our expectation;
The number is non-increasing with respect to $\delta$
and saturated for $\delta\ge30$.
In other words, when $\delta$ is larger, 
more varieties of initial solutions are generated
in a given number of iterations.

However, setting $\delta$ to a large value does not necessarily
lead to discovery of better solutions.
The dashed line in \figref{delay} shows the averaged number of iterations 
that ILPS takes to find an optimal solution;
we know that the optimal size is 24
since we solve the instance optimally by CPLEX. 
When $\delta\le40$,
the number is approximately decreasing
and takes the minimum at $\delta=40$,
but a larger $\delta$ does not make any improvement. 
Hence, given an instance,
we need to choose an appropriate value of $\delta$ carefully.

\subsection{Performance Validation}
\label{sec:comp.perform}

We run ILPS algorithm for 80 DIMACS instances that are
downloadable from \cite{DIMACS}. 
We generate the first initial solution $S$ by the maximum-degree greedy method,
and fix the parameter $\nu$ to three. 
For $(k,\delta)$, all pairs in 
 $\{2,3\}\times\{2^0,\dots,2^6\}$ are tested. 
For each instance and each $(k,\delta)$,
we run ILPS algorithm 10 times, using different random seeds. 
We terminate the algorithm by the time limit.
The time limit is set to 200\,s. 
When $k=3$, we modify \algref{ILSPS} so that 
$\textsc{PlateauSearch}(S,k)$ in \lineref{plateau}
is called only when $|S|\le|S^\ast|+2$
as the plateau search is rather time-consuming. 

We take four competitors from \cite{WCSY.2017} and \cite{WLZY.2017}.
The first is MEM, a tabu-search based memetic algorithm in \cite{WCSY.2017}.
The second is GP, the GRASP+PC algorithm in \cite{WLZY.2017}. 
The third is CP, which stands for CPLEX12.6~\cite{CPLEX} that solves
an integer optimization model of the MinIDS problem. 
The fourth is LS, which stands for LocalSolver5.5~\cite{LSSOL},
a general discrete optimization solver based on local search. 
MEM is run on a computer with a 2.0GHz CPU and a 4GB memory,
whereas the other competitors are run on computers with a 2.3GHz CPU
and an 8GB memory. 
The time limit of MEM and GP is set to 200\,s,
and that of CP and LS is set to 3600\,s. 

In \tabref{dimacs.sel},
we show the results on selected instances. 
The columns ``$n$'' and ``$p$'' indicate the number of vertices
and the edge density, respectively. 
The edge density is between 0.1 and 0.5
in all instances except hamming8-2.
In our context, the instances are expected to be difficult. 
For ILPS, we show the results for $(k,\delta)=(2,2^6)$ in detail,
regarding this pair as the representative. 
The columns ``Min'' and ``Max'' indicate
the minimum/maximum solution size over 10 runs,
and the column ``Avg'' indicates the average.
The column ``TTB'' indicates the time to best (in seconds),
that is, the average of the computation time
that ILPS takes to find the solution of the size ``Min''. 
The symbol $\varepsilon$ represents that
the time is less than 0.1\,s.  
The column ``Best'' indicates the minimum solution size attained
over all $(k,\delta)\in\{2,3\}\times\{2^0,\dots,2^6\}$. 
The rightmost four columns indicate the solution size
attained by the competitors. 
The symbol $\ast$ before the instance name
indicates that the solution size attained by CPLEX is optimal.

\begin{table}[t!]
  \centering
  \caption{Selected results from the validation experiments on DIMACS graphs}
  \label{tab:dimacs.sel}
  \begin{tabular}{lrr|rrrr|r|r|rrr}
\hline
		&	\multicolumn{1}{c}{$n$}	&	\multicolumn{1}{c|}{$p$}	&	\multicolumn{5}{c|}{ILPS}									&	\multicolumn{1}{c|}{\cite{WCSY.2017}}	&	\multicolumn{3}{c}{\cite{WLZY.2017}}					\\
\cline{4-12}																								
		&		&		&	\multicolumn{4}{c|}{($k=2$, $\nu=2^6$)}							&		&	\multicolumn{1}{c|}{MEM}	&	\multicolumn{1}{c}{GP}	&	\multicolumn{1}{c}{CP}	&	\multicolumn{1}{c}{LS}	\\
		&		&		&	Min	&	Avg	&	Max	&	TTB	&	Best	&		&		&		&		\\
\hline																								
	~\,brock400\_2	&	400	&	.25	&	10	&	10.0	&	10	&	1.1	&	9	&	9	&	10	&	10	&	11	\\
	~\,C1000.9	&	1000	&	.10	&	27	&	27.8	&	29	&	0.0	&	\bf 26	&	27	&	27	&	29	&	30	\\
	$\ast$C125.9	&	125	&	.10	&	14	&	14.0	&	14	&	0.1	&	14	&	14	&	15	&	14	&	14	\\
	~\,C2000.9	&	2000	&	.10	&	\bf 32	&	33.6	&	35	&	12.1	&	\bf 32	&	33	&	33	&	48	&	36	\\
	~\,C4000.5	&	4000	&	.50	&	\bf 7	&	7.9	&	8	&	49.7	&	\bf 7	&	8	&	8	&	-	&	-	\\
	~\,C500.9	&	500	&	.10	&	22	&	22.2	&	23	&	92.3	&	\bf 21	&	22	&	23	&	23	&	22	\\
	~\,gen400\_p0.9\_55	&	400	&	.10	&	20	&	20.1	&	21	&	39.4	&	20	&	20	&	21	&	22	&	22	\\
	~\,gen400\_p0.9\_65	&	400	&	.10	&	20	&	20.7	&	21	&	99.0	&	20	&	20	&	21	&	21	&	22	\\
	$\ast$hamming8-2	&	256	&	.03	&	36	&	36.0	&	36	&	0.0	&	32	&	-	&	32	&	32	&	32	\\
	~\,keller6	&	3361	&	.18	&	18	&	18.0	&	18	&	26.1	&	\bf 16	&	18	&	18	&	32	&	19	\\
	$\ast$san200\_0.7\_1	&	200	&	.30	&	6	&	6.1	&	7	&	85.9	&	6	&	6	&	7	&	6	&	7	\\
	$\ast$san200\_0.9\_1	&	200	&	.10	&	15	&	15.0	&	15	&	16.7	&	15	&	15	&	16	&	15	&	16	\\
	~\,san400\_0.7\_3	&	400	&	.30	&	7	&	7.8	&	8	&	106.6	&	7	&	7	&	8	&	8	&	9	\\
\hline
\end{tabular}
\end{table}

The table contains only results on the 13 selected instances
such that the solutions sizes attained by ``Best'', ``MEM'' and ``GP''
are not-all-equal, except hamming8-2.
We guarantee that, for the remaining 67 instances,
ILPS's ``Best'' is as good as any competitor. 
The boldface indicates
that the solution size is strictly smaller than
those of the competitors. 
Then we update the best-known solution size in five graphs. 
These show the effectiveness of the proposed local search
and the ILPS algorithm. 
All results are included in the appendix.

\invis{
With $(k,\delta)=(2,2^6)$,
ILPS finds a solution whose size
is no more than the three competitors
for all instances except hamming8-2.
}
For hamming8-2, when $k=2$, ILPS 
cannot find a solution of the optimal size 32
for any penalty delay $\delta\in\{2^0,\dots,2^6\}$.
However, when $k=3$, ILPS finds an optimal solution
with $\delta=2^0$, $2^1$ and $2^2$.

Before closing this section, let us report our preliminary results briefly.
\begin{itemize}
\item A preliminary version of ILPS happened to find
  a solution of the size 31 for C2000.9
  and a solution of the size 15 for keller6. 
  See the detail for the appendix. 
\item Let us consider a finer swap operation, {\em $(j,k)$-swap\/},
  that obtains another solution by dropping exactly $k$ vertices
  from the current one
  and then by adding exactly $j$ vertices to it.
  One can prove that, given a solution $S$ and a constant $k$,
  we can improve $S$ by $(1,k)$-swap or conclude that it is not possible
  in $O(n\Delta)$ time.
  We implemented $(1,k)$-swap in a preliminary version of ILPS,
  but it does not yield significant improvement even when $k$ is set to a constant larger than three. 
  
\item We tested Laforest and Phan's exact algorithm~\cite{LP.2013},
  and found that the algorithm is not suitable for a task
  of finding a good solution quickly. 
  The source code is available at
  \url{http://todo.lamsade.dauphine.fr/spip.php?article42}.
  \invis{
  As the original program does not have a function to terminate the algorithm
  by a given time limit, we append it to the source code by ourselves.
  The running time of the modified program
  is approximately 1.5 times that of the original program.
  Then, setting the time limit to 300s,
  we run the program for 80 DIMACS graphs.
  It does not work for large graphs such that
  the number of vertices or edges is no less than $2^{16}$. 
  The results are not good, and we claim that
  }
\item BHOSLIB~\cite{BHOSLIB} is another well-known collection
  of benchmark instances.
  It contains 36 instances such that
  $n$ is between 450 and 4000 and that $p$ is no less than 0.82. 
  Hence, the BHOSLIB instances are expected to be easy in our context.
  The ILPS with $(k,\delta)=(2,2^6)$
  finds a solution of the size three for all the instances.
  We also run CPLEX12.8 for 200\,s,
  generating an initial solution by
  the maximum-degree greedy algorithm.
  CPLEX12.8 finds a solution of the size five for frb100-40,
  and a solution of the size three for the other instances. 
  In addition, the solution of the size three is proved to be optimal 
  for 15 instances whose names start with frb30, frb35 and frb40. 
\end{itemize}

\section{Concluding Remark}
\label{sec:conc}

We have considered an efficient local search
for the MinIDS problem.
We proposed fast $k$-neighborhood search algorithms
for $k=2$ and 3, and developed a metaheuristic algorithm named ILPS
that repeats the local search and the plateau search iteratively. 
ILPS is so effective that it updates the best-known solution size
in five DIMACS graphs. 

The proposed local search is applicable to other metaheuristics
such as genetic algorithms,
as a key tool of local improvement. 
The future work includes an extension of the local search 
to a weighted version of the MinIDS problem.

\newpage
\bibliographystyle{plainurl}
\bibliography{mybib}

\newpage
\appendix
\section*{Appendix}
\subsection*{Proof of \thmref{ls.3}}
For preparation, let us introduce the following proposition.
This is a generalization of \propref{adj}
in the sense that the vertex subset is taken arbitrarily. 
The proof is similar to \propref{adj}. 
\begin{prop}
  \label{prop:adjgen}
  Given an arbitrary vertex subset $F$ and a vertex $v\in F$,
  we can decide whether $v$ is adjacent to all vertices in $F\setminus\{v\}$ in $O(|F|+\deg(v))$ time. 
\end{prop}
\begin{proof}
  We let every $v\in V$ have another integral variable,
  which we denote by $\chi(v)$. 
  Initially, $\chi(v)$ is set to zero. 
  We also maintain a global integral variable $\gamma$ 
  that is set to one initially. 

  First, for all $v\in F$, we set $\chi(v)$ to the current $\gamma$
  (i.e., $\chi(v)\gets\gamma$). 
  We then count the number of vertices $v$ in $N(v_0)$
  such that $\chi(v)=\gamma$. 
  If the number equals to $|F|-1$, then we can regard
  that $v_0$ is adjacent to all vertices in $F\setminus\{v_0\}$. 
  As postprocessing, we increase $\gamma$ by one
  (i.e., $\gamma\gets\gamma+1$). 
\end{proof}

An integral variable is bounded in conventional programming languages.
If $\gamma$ reaches the upper limit (e.g., INT\_MAX in C),
then we reset $\chi(v)$ to zero for all $v\in V$
and $\gamma$ to one again.

We prove \thmref{ls.3} in \secref{ls.3}. 
If there is an improved solution $(S\setminus D)\cup A$,
then the four situations from (i) to (iv) are possible
as to the tightnesses of vertices in $A$. 
(For illustration, see \figref{N3}.)
Given a non-solution vertex $a$ in (i) to (iv),
the following Lemmas~\ref{lem:ls.3-1} to \ref{lem:ls.3-4}
show time complexities of
finding an improved solution or concluding that
no such solution exists, respectively.  

\begin{lem}
  \label{lem:ls.3-1}
  Suppose that a 3-tight vertex $a$ is given. 
  Let $D=\{x,y,z\}$ be the set of solution neighbors of $a$.
  We can decide whether $(S\setminus D)\cup\{a\}$ is a
  solution or not in $O(\Delta)$ time. 
\end{lem}
\begin{proof}
The set $D$ can be decided in $O(\deg(a))$ time. 
It suffices to decide whether $a$ is adjacent to 
all vertices in $F(D)\setminus\{a\}$. 
This can be done in $O(\Delta)$ time from \propref{adj}. 
\end{proof}

\begin{lem}
  \label{lem:ls.3-2}
  Suppose that a 3-tight vertex $a$ is given. 
  Let $D=\{x,y,z\}$ be the set of solution neighbors of $a$.
  We can find a non-solution vertex $b\in F(D)$ 
  such that $(S\setminus D)\cup\{a,b\}$ is a solution
  or conclude that such $b$ does not exist in $O(\Delta^2)$ time. 
\end{lem}
\begin{proof}
  Let $F$ be the subset of $F(D)$ 
  such that the vertices in $F$ are not adjacent to $a$. 
  The sets $D$, $F(D)$ and $F$ can be constructed in $O(\Delta)$ time. 
  All we have to do is to check whether there is $b\in F$
  such that $b$ is adjacent to all vertices in $F\setminus\{b\}$. 
  From \propref{adjgen} and $|F|=O(\Delta)$, 
  this can be done in $O(\Delta^2)$ time. 
\end{proof}

\begin{lem}
  \label{lem:ls.3-3}
  Suppose that a 2-tight vertex $a$ is given. 
  We can decide in $O(\Delta^2)$ time
  whether there exists a 2-tight vertex $b$ such that:
  \begin{itemize}
  \item $a$ and $b$ have exactly one solution neighbor in common;
  \item $(S\setminus D)\cup\{a,b\}$ is a solution,
    where $D=(N(a)\cup N(b))\cap S$.
  \end{itemize}
\end{lem}
\begin{proof}
Let $D_a=\{x,y\}$ be the set of solution neighbors of $a$. 
The target 2-tight vertex $b$ 
should be a neighbor of either $x$ or $y$,
but not both, as $a$ and $b$
have exactly one solution neighbor in common. 
Hence, there are at most $\deg(x)+\deg(y)\le2\Delta$
candidates for $b$. 

For each candidate of $b$, let $D_b=\{x,z\}$ be the set of solution neighbors of $b$. 
(If $z=y$, then we discard this candidate.)
Let $D=D_a\cup D_b$. 
To check whether $(S\setminus D)\cup\{a,b\}$ is a solution,
it suffices to check
whether $\{a,b\}$ is a solution in the subgraph $G[F(D)]$. 
It takes $O(\Delta)$ time to decide $D_b$, 
to decide whether $\{a,b\}$ is independent,
and decide whether $\{a,b\}$ is dominating
the vertices in $F(D)$. 
\end{proof}

\begin{lem}
  \label{lem:ls.3-4}
  Suppose that a 2-tight vertex $a$ is given. 
  We can decide in $O(\Delta^3)$ time
  whether there exist a 2-tight vertex $a'$
  and a 1-tight vertex $b$ such that:
  \begin{itemize}
    \item $a$ and $a'$ are adjacent,
      and have exactly one solution neighbor in common.
      Let $D_a=\{x,y\}$ and $D_{a'}=\{x,z\}$ be the sets
      of solution neighbors of $a$ and $a'$ respectively;
    \item the unique solution neighbor of $b$ is $z$; 
    \item $(S\setminus D)\cup\{a,b\}$ is a solution,
      where $D=\{x,y,z\}$.
  \end{itemize}
\end{lem}
\begin{proof}
  Similarly to \lemref{ls.3-3},
  there are at most $\deg(x)+\deg(y)\le2\Delta$ candidates for $a'$. 
  The adjacency between $a$ and $a'$ can be checked in $O(\Delta)$ time.

  Note that the number of candidates for $z$ is 
  also at most $2\Delta$. 
  Each candidate of $z$ has at most $\Delta$ 1-tight neighbors
  that are the candidates of $b$. 
  Hence, for $b$, there are $O(\Delta^2)$ candidates. 

  For each candidate of $b$, 
  to check whether $(S\setminus D)\cup\{a,b\}$ is a solution,
  it suffices to check
  whether $\{a,b\}$ is a solution in the subgraph $G[F(D)]$,
  which can be done in $O(\Delta)$ time. 
  Then we have the time bound $O(\Delta^3)$. 
\end{proof}

{\bf (Proof of \thmref{ls.3})}
For every 3-tight vertex $a$,
check whether there is an improved solution
in the situations (i) and (ii). 
This can be done in $O(\Delta^2)$ time
from Lemmas \ref{lem:ls.3-1} and \ref{lem:ls.3-2}. 
Similarly, for every 2-tight vertex $a$,
check whether there is an improved solution
in the situations (iii) and (iv). 
This can be done in $O(\Delta^3)$ time
from Lemmas \ref{lem:ls.3-3} and \ref{lem:ls.3-4}. 

As there are $O(n)$ non-solution vertices,
we have the time bound $O(n\Delta^3)$. 
\hfill
\qed

\newpage
\subsection*{All Computational Results on DIMACS Graphs}

The next \tabref{dimacs.all} shows all results on 80 DIMACS graphs
that are downloadable from \cite{DIMACS}. 
The column ``CP12.8'' represents CPLEX12.8. 
We run CPLEX12.8 for each instance,
setting the time limit to 200 seconds. 
An initial solution is constructed by 
the maximum-degree greedy algorithm.
A hyphen in the rightmost four columns
indicates that the corresponding result
is not available in \cite{WCSY.2017,WLZY.2017}. 

As mentioned in the paper,
we happened to find a solution of the size 31 for C2000.9
and a solution of the size 15 for keller6 
by a preliminary version of ILPS. 
The vertices in the solution for C2000.9
have the following IDs:
\begin{quote}
  23, 78, 161, 252, 279, 344, 441, 556, 662, 671, 703, 769, 847, 864, 926, 952, 1056, 1266, 1274, 1475, 1540, 1619, 1636, 1641, 1646, 1673, 1826, 1839, 1915, 1947, 1979.
\end{quote}
The solution for keller6 is
the set of vertices with the following IDs:
\begin{quote}
  169, 601, 659, 855, 1020, 1215, 1352, 1586, 2052, 2376, 2463, 2818, 2847, 2944, 3281. 
\end{quote}
By the ID of a vertex,
we mean an integer that is given to the vertex
in the DIMACS files.

{\scriptsize
\begin{longtable}[c]{p{-5mm}l|rrrr|r|r|r|rrr}
  \caption{{\normalsize All results of the validation experiments on DIMACS graphs}}
  \label{tab:dimacs.all}
  \\
  \hline
	&		&	\multicolumn{5}{c|}{ILPS}									&		&	\multicolumn{1}{c|}{\cite{WCSY.2017}}	&	\multicolumn{3}{c}{\cite{WLZY.2017}}					\\
  \cline{3-7}																							
  \cline{9-12}																							
	&		&	\multicolumn{4}{c|}{($k=2$, $\nu=2^6$)}							&		&	\multicolumn{1}{c|}{CP}	&	\multicolumn{1}{c|}{MEM}	&	\multicolumn{1}{c}{GP}	&	\multicolumn{1}{c}{CP}	&	\multicolumn{1}{c}{LS}	\\
	&		&	Min	&	Avg	&	Max	&	TTB	&	Best	&	\multicolumn{1}{c|}{12.8}	&		&		&	\multicolumn{1}{c}{12.6}	&		\\
\hline
$\ast$	&	brock200\_1	&	8	&	8.0	&	8	&	$\varepsilon$	&	8	&	8	&	-	&	-	&	-	&	-	\\
$\ast$	&	brock200\_2	&	4	&	4.0	&	4	&	0.3	&	4	&	4	&	4	&	4	&	4	&	4	\\
$\ast$	&	brock200\_3	&	5	&	5.0	&	5	&	0.7	&	5	&	5	&	-	&	-	&	-	&	-	\\
$\ast$	&	brock200\_4	&	6	&	6.0	&	6	&	0.8	&	6	&	6	&	6	&	6	&	6	&	6	\\
	&	brock400\_1	&	10	&	10.0	&	10	&	1.1	&	10	&	10	&	-	&	-	&	-	&	-	\\
	&	brock400\_2	&	10	&	10.0	&	10	&	32.2	&	9	&	10	&	9	&	10	&	10	&	11	\\
	&	brock400\_3	&	9	&	9.3	&	10	&	168.0	&	9	&	10	&	-	&	-	&	-	&	-	\\
	&	brock400\_4	&	9	&	9.8	&	10	&	102.8	&	9	&	10	&	9	&	9	&	10	&	11	\\
	&	brock800\_1	&	8	&	8.3	&	9	&	33.5	&	8	&	9	&	-	&	-	&	-	&	-	\\
	&	brock800\_2	&	8	&	8.7	&	9	&	85.9	&	8	&	9	&	8	&	8	&	9	&	9	\\
	&	brock800\_3	&	8	&	8.4	&	9	&	81.2	&	8	&	10	&	-	&	-	&	-	&	-	\\
	&	brock800\_4	&	8	&	8.5	&	9	&	$\varepsilon$	&	8	&	9	&	8	&	8	&	9	&	9	\\
$\ast$	&	c-fat200-1	&	10	&	10.0	&	10	&	$\varepsilon$	&	10	&	10	&	-	&	-	&	-	&	-	\\
$\ast$	&	c-fat200-2	&	22	&	22.0	&	22	&	$\varepsilon$	&	22	&	22	&	-	&	-	&	-	&	-	\\
$\ast$	&	c-fat200-5	&	56	&	56.0	&	56	&	$\varepsilon$	&	56	&	56	&	-	&	-	&	-	&	-	\\
$\ast$	&	c-fat500-1	&	12	&	12.0	&	12	&	$\varepsilon$	&	12	&	12	&	-	&	-	&	-	&	-	\\
$\ast$	&	c-fat500-10	&	124	&	124.0	&	124	&	$\varepsilon$	&	124	&	124	&	-	&	-	&	-	&	-	\\
$\ast$	&	c-fat500-2	&	24	&	24.0	&	24	&	$\varepsilon$	&	24	&	24	&	-	&	-	&	-	&	-	\\
$\ast$	&	c-fat500-5	&	62	&	62.0	&	62	&	$\varepsilon$	&	62	&	62	&	-	&	-	&	-	&	-	\\
	&	C1000.9	&	27	&	27.8	&	29	&	$\varepsilon$	&	26	&	30	&	27	&	27	&	29	&	30	\\
$\ast$	&	C125.9	&	14	&	14.0	&	14	&	1.1	&	14	&	14	&	14	&	15	&	14	&	14	\\
	&	C2000.5	&	7	&	7.0	&	7	&	12.1	&	7	&	8	&	7	&	7	&	11	&	8	\\
	&	C2000.9	&	32	&	33.6	&	35	&	26.5	&	32	&	33	&	33	&	33	&	48	&	36	\\
	&	C250.9	&	17	&	17.0	&	17	&	49.7	&	17	&	17	&	17	&	17	&	18	&	18	\\
	&	C4000.5	&	7	&	7.9	&	8	&	92.3	&	7	&	9	&	8	&	8	&	-	&	-	\\
	&	C500.9	&	22	&	22.2	&	23	&	0.2	&	21	&	23	&	22	&	23	&	23	&	22	\\
	&	DSJC1000.5	&	6	&	6.0	&	6	&	1.3	&	6	&	7	&	6	&	6	&	6	&	6	\\
	&	DSJC500.5	&	5	&	5.0	&	5	&	1.6	&	5	&	5	&	5	&	5	&	10	&	7	\\
$\ast$	&	gen200\_p0.9\_44	&	16	&	16.0	&	16	&	5.4	&	16	&	16	&	16	&	16	&	16	&	16	\\
$\ast$	&	gen200\_p0.9\_55	&	16	&	16.0	&	16	&	39.4	&	16	&	16	&	16	&	16	&	16	&	16	\\
	&	gen400\_p0.9\_55	&	20	&	20.1	&	21	&	99.0	&	20	&	20	&	20	&	21	&	22	&	22	\\
	&	gen400\_p0.9\_65	&	20	&	20.7	&	21	&	95.3	&	20	&	22	&	20	&	21	&	21	&	22	\\
	&	gen400\_p0.9\_75	&	20	&	20.7	&	21	&	28.6	&	20	&	21	&	20	&	20	&	21	&	22	\\
	&	hamming10-2	&	128	&	131.1	&	133	&	1.2	&	128	&	161	&	-	&	-	&	-	&	-	\\
	&	hamming10-4	&	12	&	12.0	&	12	&	$\varepsilon$	&	12	&	14	&	12	&	12	&	14	&	12	\\
$\ast$	&	hamming6-2	&	12	&	12.0	&	12	&	$\varepsilon$	&	12	&	12	&	12	&	12	&	12	&	12	\\
$\ast$	&	hamming6-4	&	2	&	2.0	&	2	&	$\varepsilon$	&	2	&	2	&	2	&	2	&	2	&	2	\\
$\ast$	&	hamming8-2	&	36	&	36.0	&	36	&	$\varepsilon$	&	32	&	32	&	-	&	32	&	32	&	32	\\
$\ast$	&	hamming8-4	&	4	&	4.0	&	4	&	$\varepsilon$	&	4	&	4	&	4	&	4	&	4	&	4	\\
$\ast$	&	johnson16-2-4	&	8	&	8.0	&	8	&	$\varepsilon$	&	8	&	8	&	8	&	8	&	8	&	8	\\
$\ast$	&	johnson32-2-4	&	16	&	16.0	&	16	&	$\varepsilon$	&	16	&	16	&	16	&	16	&	16	&	16	\\
$\ast$	&	johnson8-2-4	&	4	&	4.0	&	4	&	$\varepsilon$	&	4	&	4	&	4	&	4	&	4	&	4	\\
$\ast$	&	johnson8-4-4	&	7	&	7.0	&	7	&	$\varepsilon$	&	7	&	7	&	7	&	7	&	7	&	7	\\
$\ast$	&	keller4	&	5	&	5.0	&	5	&	8.9	&	5	&	5	&	5	&	5	&	5	&	5	\\
	&	keller5	&	9	&	9.0	&	9	&	26.1	&	9	&	11	&	9	&	9	&	11	&	10	\\
	&	keller6	&	18	&	18.0	&	18	&	$\varepsilon$	&	16	&	20	&	18	&	18	&	32	&	19	\\
$\ast$	&	MANN\_a27	&	27	&	27.0	&	27	&	$\varepsilon$	&	27	&	27	&	27	&	27	&	27	&	27	\\
$\ast$	&	MANN\_a45	&	45	&	45.0	&	45	&	$\varepsilon$	&	45	&	45	&	45	&	45	&	45	&	45	\\
$\ast$	&	MANN\_a81	&	81	&	81.0	&	81	&	$\varepsilon$	&	81	&	81	&	81	&	81	&	81	&	81	\\
$\ast$	&	MANN\_a9	&	9	&	9.0	&	9	&	0.9	&	9	&	9	&	9	&	9	&	9	&	12	\\
$\ast$	&	p\_hat1000-1	&	3	&	3.0	&	3	&	54.8	&	3	&	3	&	-	&	-	&	-	&	-	\\
	&	p\_hat1000-2	&	6	&	6.1	&	7	&	21.5	&	6	&	7	&	-	&	-	&	-	&	-	\\
	&	p\_hat1000-3	&	11	&	12.0	&	13	&	$\varepsilon$	&	11	&	13	&	-	&	-	&	-	&	-	\\
	&	p\_hat1500-1	&	4	&	4.0	&	4	&	34.7	&	4	&	4	&	-	&	-	&	-	&	-	\\
	&	p\_hat1500-2	&	8	&	8.0	&	8	&	13.3	&	7	&	9	&	-	&	-	&	-	&	-	\\
	&	p\_hat1500-3	&	14	&	14.7	&	15	&	$\varepsilon$	&	14	&	18	&	-	&	-	&	-	&	-	\\
$\ast$	&	p\_hat300-1	&	3	&	3.0	&	3	&	$\varepsilon$	&	3	&	3	&	-	&	-	&	-	&	-	\\
$\ast$	&	p\_hat300-2	&	5	&	5.0	&	5	&	0.9	&	5	&	5	&	-	&	-	&	-	&	-	\\
$\ast$	&	p\_hat300-3	&	9	&	9.0	&	9	&	$\varepsilon$	&	9	&	9	&	-	&	-	&	-	&	-	\\
$\ast$	&	p\_hat500-1	&	3	&	3.0	&	3	&	0.3	&	3	&	3	&	-	&	-	&	-	&	-	\\
$\ast$	&	p\_hat500-2	&	6	&	6.0	&	6	&	28.4	&	6	&	6	&	-	&	-	&	-	&	-	\\
	&	p\_hat500-3	&	10	&	10.0	&	10	&	0.2	&	10	&	11	&	-	&	-	&	-	&	-	\\
$\ast$	&	p\_hat700-1	&	3	&	3.0	&	3	&	182.4	&	3	&	3	&	-	&	-	&	-	&	-	\\
	&	p\_hat700-2	&	6	&	6.9	&	7	&	52.0	&	6	&	6	&	-	&	-	&	-	&	-	\\
	&	p\_hat700-3	&	11	&	11.0	&	11	&	48.1	&	11	&	12	&	-	&	-	&	-	&	-	\\
	&	san1000	&	4	&	4.8	&	5	&	85.9	&	4	&	4	&	4	&	4	&	4	&	5	\\
$\ast$	&	san200\_0.7\_1	&	6	&	6.1	&	7	&	$\varepsilon$	&	6	&	6	&	6	&	7	&	6	&	7	\\
$\ast$	&	san200\_0.7\_2	&	6	&	6.0	&	6	&	16.7	&	6	&	6	&	6	&	6	&	6	&	6	\\
$\ast$	&	san200\_0.9\_1	&	15	&	15.0	&	15	&	2.3	&	15	&	15	&	15	&	16	&	15	&	16	\\
	&	san200\_0.9\_2	&	16	&	16.0	&	16	&	20.4	&	16	&	16	&	16	&	16	&	16	&	16	\\
	&	san200\_0.9\_3	&	15	&	15.2	&	16	&	0.3	&	15	&	17	&	15	&	15	&	15	&	15	\\
$\ast$	&	san400\_0.5\_1	&	4	&	4.0	&	4	&	64.9	&	4	&	4	&	4	&	4	&	4	&	4	\\
	&	san400\_0.7\_1	&	7	&	7.9	&	8	&	96.1	&	7	&	7	&	7	&	7	&	8	&	8	\\
	&	san400\_0.7\_2	&	7	&	7.7	&	8	&	106.6	&	7	&	7	&	7	&	7	&	7	&	8	\\
	&	san400\_0.7\_3	&	7	&	7.8	&	8	&	59.2	&	7	&	8	&	7	&	8	&	8	&	9	\\
	&	san400\_0.9\_1	&	20	&	20.4	&	21	&	$\varepsilon$	&	19	&	20	&	-	&	-	&	-	&	-	\\
$\ast$	&	sanr200\_0.7	&	7	&	7.0	&	7	&	7.0	&	7	&	7	&	-	&	-	&	-	&	-	\\
	&	sanr200\_0.9	&	16	&	16.0	&	16	&	$\varepsilon$	&	16	&	16	&	-	&	-	&	-	&	-	\\
	&	sanr400\_0.5	&	5	&	5.0	&	5	&	10.6	&	5	&	5	&	-	&	-	&	-	&	-	\\
	&	sanr400\_0.7	&	8	&	8.0	&	8	&	$\varepsilon$	&	8	&	9	&	-	&	-	&	-	&	-	\\
\hline
\end{longtable}
}

\end{document}